\documentclass[a4paper,11pt]{article}
\usepackage[OT1]{fontenc}
\usepackage{amsthm,amsmath,graphicx,latexsym,amssymb,array}
\usepackage{natbib}
\usepackage[colorlinks,citecolor=blue,urlcolor=blue]{hyperref}
\usepackage{adjustbox,blindtext}

\theoremstyle{plain}
\newtheorem{theorem}{Theorem}
\newtheorem{definition}{Definition}
\newtheorem{corollary}{Corollary}

\title{Paired sample tests in infinite dimensional spaces}
\date{}
\author{\vspace{0.3in} Anirvan Chakraborty and Probal Chaudhuri}
\begin{document}

\maketitle
\vspace{-0.6in}
\begin{center}
Theoretical Statistics and Mathematics Unit, \\
Indian Statistical Institute \\
203, B. T. Road, Kolkata - 700108, INDIA. \\
emails: vanchak@gmail.com, probal@isical.ac.in
\end{center}
\vspace{0.15in}
\begin{abstract}
The sign and the signed-rank tests for univariate data are perhaps the most popular nonparametric competitors of the t test for paired sample problems. These tests have been extended in various ways for multivariate data in finite dimensional spaces. These extensions include tests based on spatial signs and signed ranks, which have been studied extensively by Hannu Oja and his coauthors. They showed that these tests are asymptotically more powerful than Hotelling's $T^{2}$ test under several heavy tailed distributions. In this paper, we consider paired sample tests for data in infinite dimensional spaces based on notions of spatial sign and spatial signed rank in such spaces. We derive their asymptotic distributions under the null hypothesis and under sequences of shrinking location shift alternatives. We compare these tests with some mean based tests for infinite dimensional paired sample data. We show that for shrinking location shift alternatives, the proposed tests are asymptotically more powerful than the mean based tests for some heavy tailed distributions and even for some Gaussian distributions in infinite dimensional spaces. We also investigate the performance of different tests using some simulated data.
\vspace{0.1in} 
\end{abstract}

\section{Introduction}
\label{1}
\indent For univariate data, two nonparametric competitors of the t test for one sample and paired sample problems are the sign test and the Wilcoxon signed-rank test. It is well known that these two tests enjoy certain robustness properties, and are more powerful than the t test when the underlying distribution has heavier tails than the Gaussian distribution. The sign and the signed-rank tests have been extended in several ways for data in $\mathbb{R}^{d}$. \citet{PS71} considered extensions based on coordinatewise signs and ranks. \citet{Rand89} and \citet{PR90} studied extensions of the sign test and the signed-rank test, respectively, using the notion of interdirections. \citet{CS93} proposed a class of multivariate extensions of the sign test based on data-driven transformations of the sample observations. \citet{HP02} considered multivariate signed-rank type tests based on interdirections and the ranks of the sample observations computed using pseudo-Mahalanobis distances. \citet{HNO94} and \citet{HMO97} considered multivariate versions of sign and signed-rank tests using simplices (see also \citet{Oja99}). \citet{MO95}, \citet{CCO98} and \citet{Mard99} used spatial signs and signed ranks (see also \citet{Oja10}) to construct extensions of sign and signed-rank tests for multivariate data. \citet{MOT97} and \citet{CCO98} studied the asymptotic efficiency of spatial sign and signed-rank tests relative to Hotelling's $T^{2}$ test. They showed that the tests based on spatial signs and signed ranks are asymptotically more powerful than the $T^{2}$ test under heavy tailed distributions like multivariate t distributions. Further, while Hotelling's $T^{2}$ test is optimal for multivariate Gaussian distributions, it was shown that as the data dimension increases, the performance of the spatial sign and the spatial signed-rank tests become closer to Hotelling's $T^{2}$ test under multivariate Gaussian distributions. \\
\indent Nowadays, we often come across data, which are curves or functions observed over an interval, and are popularly known as functional data. Such data are very different from multivariate data in finite dimensional spaces because the data dimension is much larger than the sample size, and also due to the fact that different sample observations may be observed at different sets of points in the interval. However, this type of data can be conveniently handled by viewing them as a sample in some infinite dimensional space, e.g., the space of real-valued functions defined over an interval. Many of the above-mentioned tests cannot be used to analyze such data. This is because the definitions of some of them involve hyperplanes, simplices etc. constructed using the data, and thus these tests require the data dimension to be smaller than the sample size. Some of the other tests involve inverses of covariance matrices computed from the sample, and such empirical covariance matrices are singular, when the data  dimension is larger than the sample size. \\
\indent Many of the function spaces, where functional data lie, are infinite dimensional Banach spaces. In this paper, we investigate a sign and a signed-rank type test for paired sample problems based on notions of spatial sign and spatial signed rank in such spaces. We derive the asymptotic distributions of the proposed spatial sign and signed-rank statistics under the null hypothesis as well as under suitable sequences of shrinking location shift alternatives. We also compare the asymptotic powers of these tests with some of the paired-sample mean based tests for infinite dimensional data. It is found that these tests outperform the mean based competitors, when the underlying distribution has heavy tails and also for some Gaussian distributions of the data. Some simulation studies are carried out to demonstrate the performance of different tests.

\section{Sign and Signed-rank type tests}
\label{2}
\indent The construction and the study of the paired sample sign and signed-rank type statistics for data in general Banach spaces (finite or infinite dimensional) will require several concepts and tools from functional analysis and probability theory in Banach spaces. Let ${\cal X}$ be a Banach space with norm $||\cdot||$. Let us denote its dual space by ${\cal X}^{*}$, which is the Banach space of real-valued continuous linear functions defined over ${\cal X}$.
\begin{definition}
The norm in ${\cal X}$ is said to be G\^ateaux differentiable at a nonzero ${\bf x} \in {\cal X}$ with derivative, say, ${\bf S}_{{\bf x}} \in {\cal X}^{*}$ if $\lim_{t \rightarrow 0} \ t^{-1}(||{\bf x} + t{\bf h}|| - ||{\bf x}||) = {\bf S}_{{\bf x}}({\bf h})$ for all ${\bf h} \in {\cal X}$. A Banach space ${\cal X}$ is said to be smooth if the norm in ${\cal X}$ is G\^ateaux differentiable at every nonzero ${\bf x} \in {\cal X}$.
\end{definition}
As a convention, we take ${\bf S}_{{\bf x}} = {\bf 0}$ if ${\bf x} = {\bf 0}$. Hilbert spaces and $L_{p}$ spaces for $p  \in (1,\infty)$ are smooth. For a Hilbert space ${\cal X}$, we have ${\bf S}_{{\bf x}} = {\bf x}/||{\bf x}||$. If ${\cal X}$ is a $L_{p}$ space for some $p \in (1,\infty)$, then ${\bf S}_{{\bf x}}({\bf h}) = \int sign\{{\bf x}({\bf s})\}|{\bf x}({\bf s})|^{p-1}{\bf h}({\bf s})d{\bf s}/||{\bf x}||^{p-1}$ for all ${\bf h} \in {\cal X}$.
\begin{definition}
A random element ${\bf X}$ in the Banach space ${\cal X}$ is said to be Bochner integrable if there exists a sequence $\{{\bf X}_{n}\}_{n \geq 1}$ of simple functions in ${\cal X}$ such that ${\bf X}_{n} \rightarrow {\bf X}$ {\it almost surely} and $E(||{\bf X}_{n} - {\bf X}||) \rightarrow 0$ as $n \rightarrow \infty$.
\end{definition}
It is known that the Bochner expectation of ${\bf X}$ exists if $E(||{\bf X}||) < \infty$. We refer to Chapters $4$ and $5$ in \citet{BV10} and Sect. $2$ in Chapter $3$ of \citet{AG80} for more details. Henceforth, the expectation of any Banach space valued random element will be in the Bochner sense. The spatial sign of ${\bf x} \in {\cal X}$ is given by ${\bf S}_{{\bf x}}$, and its spatial rank with respect to the distribution of a random element ${\bf X} \in {\cal X}$ is defined as $\boldsymbol{\Psi}_{{\bf x}} = E({\bf S}_{{\bf x} - {\bf X}})$. If ${\cal X}$ is a Hilbert space, then it follows that $\boldsymbol{\Psi}_{{\bf x}} = E\{({\bf x} - {\bf X})/||{\bf x} - {\bf X}||\}$. In particular, if ${\cal X} = \mathbb{R}^{d}$ equipped with the Euclidean norm, ${\bf S}_{{\bf x}}$ and $\boldsymbol{\Psi}_{{\bf x}}$ reduce to the usual multivariate spatial sign and spatial rank in $\mathbb{R}^{d}$ (see \citet{Oja10}).  \\
\indent Let $({\bf X}_{1},{\bf Y}_{1}), ({\bf X}_{2},{\bf Y}_{2}) \ldots,({\bf X}_{n},{\bf Y}_{n})$ be i.i.d. paired  observations, where the ${\bf X}_{i}$'s and the ${\bf Y}_{i}$'s take values in a smooth Banach space ${\cal X}$. Let ${\bf W}_{i} = {\bf Y}_{i} - {\bf X}_{i}$, $1 \leq i \leq n$, and define $\boldsymbol{\nu} = E({\bf S}_{{\bf W}_{1}})$. A paired sample sign statistic using spatial signs for testing the hypothesis $H_{0}^{(1)} : \boldsymbol{\nu} = {\bf 0}$ against $H_{1}^{(1)} : \boldsymbol{\nu} \neq {\bf 0}$ is defined as
\begin{equation*}
{\bf T}_{S} = n^{-1} \sum_{1 \leq i \leq n} {\bf S}_{{\bf W}_{i}}.
\end{equation*}
We reject $H_{0}^{(1)}$ when $||{\bf T}_{S}||$ is large. Next, define  $\boldsymbol{\theta} = E({\bf S}_{{\bf W}_{1} + {\bf W}_{2}})$. A paired sample signed-rank statistic using spatial signed ranks for testing the hypothesis $H_{0}^{(2)} : \boldsymbol{\theta} = {\bf 0}$ against $H_{1}^{(2)} : \boldsymbol{\theta} \neq {\bf 0}$ is given by
\begin{equation*}
{\bf T}_{SR} = 2\{n(n-1)\}^{-1} \sum_{1 \leq i < j \leq n} {\bf S}_{{\bf W}_{i} + {\bf W}_{j}}.
\end{equation*}
We reject $H_{0}^{(2)}$ for large values of $||{\bf T}_{SR}||$. Recently, a two sample Wilcoxon--Mann--Whitney type test based on spatial ranks in infinite dimensional spaces have been studied by \citet{CC14d}. Note that if ${\cal X} = \mathbb{R}$, ${\bf S}_{x} = sign(x)$. Thus, ${\bf T}_{S}$ and ${\bf T}_{SR}$ reduce to the univariate sign and signed-rank statistics if ${\cal X} = \mathbb{R}$. Moreover, if ${\cal X} = \mathbb{R}^{d}$, then ${\bf T}_{S}$ and ${\bf T}_{SR}$ are the spatial sign and signed-rank statistics for finite dimensional multivariate data studied by \citet{MO95}, \citet{MOT97} and \citet{Mard99}. Note that the hypothesis $H_{0}^{(1)} : \boldsymbol{\nu} = {\bf 0}$ (respectively, $H_{0}^{(2)} : \boldsymbol{\theta} = {\bf 0}$) is equivalent to the hypothesis that the spatial median of ${\bf W}_{1}$ (respectively, ${\bf W}_{1} + {\bf W}_{2}$) is zero. Suppose that ${\bf Y}_{1} - {\bf X}_{1}$ has a symmetric distribution about some $\boldsymbol{\eta} \in {\cal X}$, i.e., the distribution of ${\bf Y}_{1} - {\bf X}_{1} - \boldsymbol{\eta}$ and $\boldsymbol{\eta} - {\bf Y}_{1} + {\bf X}_{1}$ are the same. Then, it follows that both of $H_{0}^{(1)}$ and $H_{0}^{(2)}$ become the hypothesis $\boldsymbol{\eta} = {\bf 0}$. This holds, in particular, if ${\bf X}_{1}$ and ${\bf Y}_{1}$ are exchangeable, i.e., the distributions of $({\bf X}_{1},{\bf Y}_{1})$ and $({\bf Y}_{1},{\bf X}_{1})$ are the same. Further, if the distribution of ${\bf Y}_{1} - {\bf X}_{1}$ is symmetric and its mean exists, then both of $H_{0}^{(1)}$ and $H_{0}^{(2)}$ are equivalent to the hypothesis $E({\bf Y}_{1} - {\bf X}_{1}) = {\bf 0}$.

\subsection{Asymptotic distributions and the implementation of the tests}
\label{2.1}
\indent In order to study the asymptotic distributions of ${\bf T}_{S}$ and ${\bf T}_{SR}$, we introduce the following definitions.
\begin{definition}
A Banach space ${\cal X}$ is said to be of type $2$ if there exists a constant $b > 0$ such that for any $m \geq 1$ and independent zero mean random elements ${\bf U}_{1},{\bf U}_{2},\ldots,{\bf U}_{m}$ in ${\cal X}$ with $E(||{\bf U}_{i}||^{2}) < \infty$, $1 \leq i \leq m$, we have $E(||\sum_{i=1}^{m} {\bf U}_{i}||^{2}) \leq b\sum_{i=1}^{m} E(||{\bf U}_{i}||^{2})$.
\end{definition}
\begin{definition}
A Banach space ${\cal X}$ is said to be $p$-uniformly smooth for some $p \in (1,2]$ if for every $q \geq 1$ there exists a constant $\alpha_{q} > 0$ such that for any zero mean martingale sequence $({\bf M}_{m},{\cal G}_{m})_{m \geq 1}$ in  ${\cal X}$, we have $E(||{\bf M}_{m}||^{q}) \leq \alpha_{q}\sum_{i=1}^{m} E(||{\bf M}_{i} - {\bf M}_{i-1}||^{p})^{q/p}$. Here, the sequence $({\bf M}_{m})_{m \geq 1}$ is adapted to the filtration $({\cal G}_{m})_{m \geq 1}$.
\end{definition}
Any $2$-uniformly smooth Banach space is of type $2$. Hilbert spaces are $2$-uniformly smooth, and $L_{p}$ spaces are $\widetilde{p}$-uniformly smooth, where $\widetilde{p} = \min(p,2)$ for $p \in (1,\infty)$.
\begin{definition}
A continuous linear operator ${\bf C} : {\cal X}^{*} \rightarrow {\cal X}$ is said to be symmetric if ${\bf y}\{{\bf C}({\bf x})\} = {\bf x}\{{\bf C}({\bf y})\}$ for all ${\bf x} ,{\bf y} \in {\cal X}^{*}$. It is said to be positive if ${\bf x}\{{\bf C}({\bf x})\} > 0$ for all ${\bf x} \in {\cal X}^{*}$.
\end{definition}
\begin{definition}
A random element ${\bf Z}$ in a separable Banach space ${\cal X}$ is said to have a Gaussian distribution with mean ${\bf m} \in {\cal X}$ and covariance ${\bf C}$, which we denote by $G({\bf m},{\bf C})$, if for any ${\bf u} \in {\cal X}^{*}$, ${\bf u}({\bf Z})$ has a Gaussian distribution on $\mathbb{R}$ with mean ${\bf u}({\bf m})$ and variance ${\bf u}\{{\bf C}({\bf u})\}$. Here, ${\bf C} : {\cal X}^{*} \rightarrow {\cal X}$ is a symmetric positive continuous linear operator.
\end{definition}
We refer to Sect. $7$ of Chapter $3$ in \citet{AG80}, \citet{Boro91} and Sect. 2.4 in Chapter IV of \citet{VTC87} for further details. Define ${\boldsymbol{\Pi}}_{1} : {\cal X}^{**} \rightarrow {\cal X}^{*}$ and ${\boldsymbol{\Pi}}_{2} : {\cal X}^{**} \rightarrow {\cal X}^{*}$ as
\begin{eqnarray*}
&& {\boldsymbol{\Pi}}_{1}({\bf f}) = E[{\bf f}({\bf S}_{{\bf W}_{1}}){\bf S}_{{\bf W}_{1}}] - \{{\bf f}(\boldsymbol{\nu})\}\boldsymbol{\nu}, \ \ \ \mbox{and} \\
&& {\boldsymbol{\Pi}}_{2}({\bf f}) = 4(E[{\bf f}\{E({\bf S}_{{\bf W}_{1} + {\bf W}_{2}}|{\bf W}_{1})\}E({\bf S}_{{\bf W}_{1} + {\bf W}_{2}}|{\bf W}_{1})] - \{{\bf f}(\boldsymbol{\theta})\}\boldsymbol{\theta}),
\end{eqnarray*}
where ${\bf f} \in {\cal X}^{**}$. So, ${\boldsymbol{\Pi}}_{1}$ and ${\boldsymbol{\Pi}}_{2}$ are symmetric positive continuous linear operators. For Banach space valued random elements ${\bf U}$ and ${\bf V}$ defined on the same probability space with ${\bf U}$ having finite Bochner expectation, the conditional expectation of ${\bf U}$ given  ${\bf V}$ exists and can be properly defined (see \citet[pp. $125-128$]{VTC87}). Let $({\bf X}_{i},{\bf Y}_{i})$, $1 \leq i \leq n$, be i.i.d. paired observations with the ${\bf X}_{i}$'s and the ${\bf Y}_{i}$'s taking values in a smooth Banach space ${\cal X}$. The next theorem gives the asymptotic distributions of ${\bf T}_{S}$ and ${\bf T}_{SR}$.
\begin{theorem} \label{thm1}
Suppose that the dual space ${\cal X}^{*}$ is a separable and type $2$ Banach space. Then, for any probability measure $P$ on ${\cal X}$, $n^{1/2}({\bf T}_{S} - \boldsymbol{\nu})$ converges {\it weakly} to a Gaussian limit $G({\bf 0},{\boldsymbol{\Pi}}_{1})$ as $n \rightarrow \infty$. Further, if ${\cal X}^{*}$ is $p$-uniformly smooth for some $p \in (4/3,2]$, we have weak convergence of $n^{1/2}({\bf T}_{SR} - \boldsymbol{\theta})$ to a Gaussian limit $G({\bf 0},{\boldsymbol{\Pi}}_{2})$ as $n \rightarrow \infty$.
\end{theorem}

\begin{proof}
Using the central limit theorem for i.i.d. random elements in a separable and type 2 Banach space (see \citet[Thm. $7.5$(i)]{AG80}), we get that $n^{1/2}({\bf T}_{S} - \boldsymbol{\nu})$ converges {\it weakly} to $G({\bf 0},{\boldsymbol{\Pi}}_{1})$. \\
\indent Note that ${\bf T}_{SR} - \boldsymbol{\theta}$ is a Banach space valued $U$-statistic with kernel ${\bf h}({\bf w}_{i},{\bf w}_{j}) = {\bf S}_{{\bf w}_{i} + {\bf w}_{j}} - \boldsymbol{\theta}$, which satisfies $E\{{\bf h}({\bf W}_{i},{\bf W}_{j})\} = {\bf 0}$. By the Hoeffding type decomposition for Banach space valued $U$-statistics (see \citet[p. 430]{Boro91}), we have
\begin{eqnarray*}
 {\bf T}_{SR} - \boldsymbol{\theta} = \frac{2}{n} \sum_{i=1}^{n} [E\{{\bf S}_{{\bf W}_{i} + {\bf W}'}\mid{\bf W}_{i}\} - \boldsymbol{\theta}] + {\bf R}_{n},
\end{eqnarray*}
where ${\bf W}'$ is an independent copy of ${\bf W}_{1}$. So, ${\bf R}_{n} = 2[n(n-1)]^{-1} \sum_{1 \leq i < j \leq n} \widetilde{{\bf h}}({\bf W}_{i},{\bf W}_{j})$, where $\widetilde{{\bf h}}({\bf w}_{i},{\bf w}_{j}) = {\bf h}({\bf w}_{i},{\bf w}_{j}) - E\{{\bf h}({\bf W}_{i},{\bf W}_{j})\mid{\bf W}_{i}={\bf w}_{i}\} - E\{{\bf h}({\bf W}_{i},{\bf W}_{j})\mid{\bf W}_{j}={\bf w}_{j}\}$. Note that $||\widetilde{{\bf h}}({\bf w}_{i},{\bf w}_{j})|| \leq 4$. Using the boundedness of $\widetilde{{\bf h}}(\cdot,\cdot)$ and Thm. 5.1 in \citet{Boro91}, it follows that for any $q \in [1,p]$,
\begin{equation}
E(||n^{1/2}{\bf R}_{n}||^{q}) \leq \widetilde{\alpha}_{q}n^{2-(3q/2)}  \label{thm1-eq1}
\end{equation}
for every $n \geq 2$ and a constant $\widetilde{\alpha}_{q}$. Thus, if $p > 4/3$, $E(||n^{1/2}{\bf R}_{n}||^{q})$ converges to zero as $n \rightarrow \infty$ for any $q \in (4/3,p]$. This implies that $n^{1/2}{\bf R}_{n}$ converges to zero {\it in probability} as $n \rightarrow \infty$. \\
\indent Now, $n^{-1/2} \sum_{i=1}^{n} E\{{\bf h}({\bf W}_{i},{\bf W}')\mid{\bf W}_{i}\}$ converge {\it weakly} to $G({\bf 0},{\boldsymbol{\Pi}}_{2})$ as $n \rightarrow \infty$ by the central limit theorem for i.i.d. random elements  in a separable type $2$ Banach space (see \citet[Thm. $7.5$(i)]{AG80}). This, together with the fact that $n^{1/2}{\bf R}_{n}$ converges to zero {\it in probability}, completes the proof.
\end{proof}

\indent Let $c_{1\alpha}$ and $c_{2\alpha}$ denote the $(1-\alpha)$ quantiles of the distributions of $||G({\bf 0},{\boldsymbol{\Pi}}_{1})||$ and $||G({\bf 0},{\boldsymbol{\Pi}}_{2})||$, respectively. The test based on ${\bf T}_{S}$ rejects $H_{0}^{(1)}$ when $||n^{1/2}{\bf T}_{S}|| > c_{1\alpha}$ and the test based on ${\bf T}_{SR}$ rejects $H_{0}^{(2)}$ when $||n^{1/2}{\bf T}_{SR}|| > c_{2\alpha}$.
\begin{corollary}
The asymptotic sizes of these tests based on ${\bf T}_{S}$ and ${\bf T}_{SR}$ will be the same as their nominal level. Further, these tests are consistent whenever $\boldsymbol{\nu} \neq {\bf 0}$ and $\boldsymbol{\theta} \neq {\bf 0}$, respectively. So, if the distribution of ${\bf Y}_{1} - {\bf X}_{1}$ is symmetric and its spatial median is non-zero, then these tests are consistent. In particular, these tests are consistent for location shift alternatives.
\end{corollary}

\indent We next describe how to compute the critical value of the tests based on ${\bf T}_{S}$ and ${\bf T}_{SR}$ using their asymptotic distributions under the null hypotheses obtained in Thm. \ref{thm1}. Suppose that ${\cal X}$ is a separable Hilbert space and ${\bf Z}$ is a zero mean Gaussian random element in ${\cal X}$ with covariance operator ${\bf C}$. Then, it can be shown using the spectral decomposition of the compact self-adjoint operator ${\bf C}$ (see Thm. IV$.2.4$, and Thm. $1.3$ and Cor. $2$ in pp. 159-160 in \citet{VTC87}) that the distribution of $||{\bf Z}||^{2}$ is a weighted sums of independent chi-square variables each with one degree of freedom, where the weights are the eigenvalues of the ${\bf C}$. Thus, if ${\cal X}$ is a separable Hilbert space, the asymptotic distributions of $||n^{1/2}{\bf T}_{S}||^{2}$ and $||n^{1/2}{\bf T}_{SR}||^{2}$ are weighted sums of independent chi-square variables each with one degree of freedom, where the weights are the eigenvalues of ${\boldsymbol{\Pi}}_{1}$ and ${\boldsymbol{\Pi}}_{2}$, respectively. The eigenvalues of ${\boldsymbol{\Pi}}_{1}$ and ${\boldsymbol{\Pi}}_{2}$ can be estimated by the eigenvalues of $\widehat{\Pi}_{1}$ and $\widehat{\Pi}_{2}$, which are defined as
\begin{eqnarray*}
\widehat{{\boldsymbol{\Pi}}_{1}} &=& \frac{1}{n-1}\left\{\sum_{i=1}^{n} \left(\frac{{\bf W}_{i}}{||{\bf W}_{i}||} - \widehat{\boldsymbol{\nu}} \right) \otimes \left(\frac{{\bf W}_{i}}{||{\bf W}_{i}||} - \widehat{\boldsymbol{\nu}} \right) \right\}, \\
\widehat{{\boldsymbol{\Pi}}_{2}} &=& \frac{4}{n-1}\left\{\sum_{i=1}^{n} \left(\frac{1}{n-1} \sum_{\substack{j=1 \\ j \neq i}}^{n} \frac{{\bf W}_{i} + {\bf W}_{j}}{||{\bf W}_{i} + {\bf W}_{j}||} - \widehat{\boldsymbol{\theta}}\right) \otimes \left(\frac{1}{n-1} \sum_{\substack{j=1 \\ j \neq i}}^{n} \frac{{\bf W}_{i} + {\bf W}_{j}}{||{\bf W}_{i} + {\bf W}_{j}||} - \widehat{\boldsymbol{\theta}}\right) \right\}.
\end{eqnarray*}
Here, ${\bf x} \otimes {\bf x} : {\cal X} \rightarrow {\cal X}$ is the tensor product in the Hilbert space ${\cal X}$, which is defined as $\langle({\bf x} \otimes {\bf x})({\bf f}),{\bf g}\rangle = \langle {\bf x},{\bf f}\rangle\langle {\bf x},{\bf g}\rangle$ for ${\bf f}, {\bf g}, {\bf x} \in {\cal X}$. Further, $\widehat{\boldsymbol{\nu}} = n^{-1} \sum_{i=1}^{n} {\bf W}_{i}/||{\bf W}_{i}||$ and $\widehat{\boldsymbol{\theta}} = 2\{n(n-1)\}^{-1} \sum_{i=1}^{n-1} \sum_{j=i+1}^{n} ({\bf W}_{i} + {\bf W}_{j})/||{\bf W}_{i} + {\bf W}_{j}||$. The critical values $c_{1\alpha}$ and $c_{2\alpha}$ can be obtained by simulating from the estimated asymptotic distributions of ${\bf T}_{S}$ and ${\bf T}_{SR}$. On the other hand, if ${\cal X}$ is a general Banach space satisfying the assumptions of Thm. \ref{thm1}, we no longer have the weighted chi-square representations for the asymptotic distributions of ${\bf T}_{S}$ and ${\bf T}_{SR}$ under the null hypotheses. However, we can estimate ${\boldsymbol{\Pi}}_{1}$ and ${\boldsymbol{\Pi}}_{2}$ by their empirical counterparts, which are defined in a similar way as the definitions above. We can simulate from the asymptotic Gaussian distributions with the estimated covariance operators to compute the critical values of the tests.

\section{Asymptotic distributions under shrinking alternatives}
\label{3}
\indent Consider i.i.d. paired observations $({\bf X}_{i},{\bf Y}_{i})$, $1 \leq i \leq n$, where the ${\bf X}_{i}$'s and the ${\bf Y}_{i}$'s take values in a smooth Banach space ${\cal X}$. In this section, we shall derive the asymptotic distribution of ${\bf T}_{S}$ and ${\bf T}_{SR}$ under sequences of shrinking location shift alternatives, where ${\bf W}_{i} = {\bf Y}_{i} - {\bf X}_{i}$ is symmetrically distributed about $n^{-1/2}\boldsymbol{\eta}$ for some fixed nonzero $\boldsymbol{\eta} \in {\cal X}$ and $1 \leq i \leq n$. For some of the finite dimensional multivariate extensions of the sign and the signed-rank tests, such alternative hypotheses have been shown to be contiguous to the null and yields nondegenerate asymptotic distributions of the test statistics (see, e.g., \citet{Rand89,MOT97,Oja99}).
\begin{definition}
The norm in ${\cal X}$ is said to be twice G\^ateaux differentiable at ${\bf x} \neq {\bf 0}$ with Hessian (or second order G\^ateaux derivative) ${\bf H}_{{\bf x}}$, which is a continuous linear map from ${\cal X}$ to ${\cal X}^{*}$, if
$\lim_{t \rightarrow 0} \ t^{-1}({\bf S}_{{\bf x} + t{\bf h}} - {\bf S}_{{\bf x}}) = {\bf H}_{{\bf x}}({\bf h})$
for all ${\bf h} \in {\cal X}$. Here, the limit is assumed to exist in the norm topology of ${\cal X}^{*}$.
\end{definition}
Norms in Hilbert spaces and $L_{p}$ spaces for $p \in [2,\infty)$ are twice G\^ateaux differentiable. We refer to Chapters $4$ and $5$ in \citet{BV10} for further details. For the next theorem, let us assume that the norm in ${\cal X}$ is twice G\^ateaux differentiable at every ${\bf x} \neq {\bf 0}$, and denote the Hessians of the functions ${\bf x} \mapsto E\{||{\bf x} - {\bf W}_{1}|| - ||{\bf W}_{1}||\}$ and ${\bf x} \mapsto E\{||2{\bf x} - {\bf W}_{1} - {\bf W}_{2}|| - ||{\bf W}_{1} + {\bf W}_{2}||\}$ by ${\bf J}^{(1)}_{{\bf x}}$ and ${\bf J}^{(2)}_{{\bf x}}$, respectively. The following theorem gives the asymptotic distributions of ${\bf T}_{S}$ and ${\bf T}_{SR}$ under the sequence of shrinking alternatives mentioned earlier.
\begin{theorem}  \label{thm2}
Suppose that ${\cal X}^{*}$ is a separable and type $2$ Banach space. Assume that the distribution of ${\bf W}_{1}$ is nonatomic, and ${\bf J}^{(1)}_{{\bf 0}}$ exists. Then, $n^{1/2}{\bf T}_{S}$ converges {\it weakly} to a Gaussian limit $G\{{\bf J}^{(1)}_{{\bf 0}}(\boldsymbol{\eta}),{\boldsymbol{\Pi}}_{1}\}$ as $n \rightarrow \infty$. Further, if ${\cal X}^{*}$ is a $p$-uniformly smooth Banach space for some $p \in (4/3,2]$ and ${\bf J}^{(2)}_{{\bf 0}}$ exists, we have weak convergence of $n^{1/2}{\bf T}_{SR}$ to a Gaussian limit $G\{{\bf J}^{(1)}_{{\bf 0}}(\boldsymbol{\eta}),{\boldsymbol{\Pi}}_{2}\}$. Here, the expectations in the definitions of ${\bf J}^{(1)}_{{\bf 0}}$ and ${\bf J}^{(2)}_{{\bf 0}}$ are taken with respect to the symmetric distribution of ${\bf W}_{1}$ about zero under the null hypothesis.
\end{theorem}

\begin{proof}
\indent We first derive the asymptotic distribution of ${\bf T}_{SR}$. Let $\boldsymbol{\eta}_{n} = n^{-1/2}\boldsymbol{\eta}$. Applying the Hoeffding type decomposition for Banach space valued $U$-statistics as in the proof of Thm. $1$, it follows that
\begin{equation}
{\bf T}_{SR} - \boldsymbol{\theta}(\boldsymbol{\eta}_{n}) = \frac{2}{n} \sum_{i=1}^{n} \{E({\bf S}_{{\bf W}_{i} + {\bf W}'}\mid{\bf W}_{i}) - \boldsymbol{\theta}(\boldsymbol{\eta}_{n})\} + \widetilde{{\bf R}}_{n},  \label{thm2-eq1}
\end{equation}
where ${\bf W}'$ is an independent copy of ${\bf W}$. Arguing as in the proof of Thm. $1$, it can be shown that $E(||n^{1/2}\widetilde{{\bf R}}_{n}||^{q})$ satisfies the bound obtained in \eqref{thm1-eq1} for every $n \geq 2$ and any $q \in (4/3,p]$. Thus, $n^{1/2}\widetilde{{\bf R}}_{n} \rightarrow 0$ {\it in probability} as $n \rightarrow \infty$ under the sequence of shrinking shifts. \\
\indent By definition, $\boldsymbol{\theta}(\boldsymbol{\eta}_{n}) = E({\bf S}_{{\bf W}_{1}^{0} + {\bf W}_{2}^{0} + 2\boldsymbol{\eta}_{n}}) = E({\bf S}_{2\boldsymbol{\eta}_{n} - {\bf W}_{1}^{0} - {\bf W}_{2}^{0}})$, where ${\bf W}_{i}^{0} = {\bf W}_{i} - \boldsymbol{\eta}_{n}$ has the same distribution as that of ${\bf W}_{i}$ under the null hypothesis for $1 \leq i \leq n$. So, $n^{1/2}\boldsymbol{\theta}(\boldsymbol{\eta}_{n})$ converges to ${\bf J}_{{\bf 0}}^{(2)}(\boldsymbol{\eta})$ as $n \rightarrow \infty$. \\
\indent Define ${\boldsymbol{\Phi}}_{n}({\bf W}_{i}) = 2n^{-1/2}\{E({\bf S}_{{\bf W}_{i} + {\bf W}'}\mid{\bf W}_{i}) - \boldsymbol{\theta}(\boldsymbol{\eta}_{n})\}$. So, $E\{{\boldsymbol{\Phi}}_{n}({\bf W}_{i})\} = {\bf 0}$. To prove the asymptotic Gaussianity of $\sum_{i=1}^{n} {\boldsymbol{\Phi}}_{n}({\bf W}_{i})$, it is enough to show that the triangular array $\{{\boldsymbol{\Phi}}_{n}({\bf W}_{1}),{\boldsymbol{\Phi}}_{n}({\bf W}_{2}),\ldots,{\boldsymbol{\Phi}}_{n}({\bf W}_{n})\}_{n=1}^{\infty}$ of rowwise i.i.d. random elements satisfy the conditions of Corollary $7.8$ in \citet{AG80}. \\
\indent Observe that for any $\epsilon > 0$,
\begin{eqnarray*}
 \sum_{i=1}^{n} P\{||{\boldsymbol{\Phi}}_{n}({\bf W}_{i})|| > \epsilon\} \leq \frac{8}{({\epsilon}n)^{3/2}}\sum_{i=1}^{n} E\{||E({\bf S}_{{\bf W}_{i} + {\bf W}'}\mid {\bf W}_{i}) - \boldsymbol{\theta}(\boldsymbol{\eta}_{n})||^{3}\} \leq \frac{64}{(\epsilon^{3}n)^{1/2}}.
\end{eqnarray*}
Thus, $\lim_{n \rightarrow \infty} \sum_{i=1}^{n} P\{||{\boldsymbol{\Phi}}_{n}({\bf W}_{i})|| > \epsilon\} = 0$ for every $\epsilon > 0$, which ensures that condition ($1$) of Corollary $7.8$ in \citet{AG80} holds. \\
\indent We next verify condition ($2$) of Corollary $7.8$ in \citet{AG80}. Let us fix ${\bf f} \in {\cal X}$. Since $||{\bf S}_{{\bf x}}|| = 1$ for all ${\bf x} \neq {\bf 0}$, we can choose $\delta = 1$ in that condition ($2$). Then, using the linearity of ${\bf f}$, we have
\begin{eqnarray}
\sum_{i=1}^{n} E[{\bf f}^{2}\{{\boldsymbol{\Phi}}_{n}({\bf W}_{i})\}] = 4n^{-1} \sum_{i=1}^{n} E[\{V_{n,i} - E(V_{n,i})\}^{2}],   \label{thm2-eq2}
\end{eqnarray}
where $V_{n,i} = {\bf f}\{E({\bf S}_{{\bf W}_{i} + {\bf W}'}\mid{\bf W}_{i})\}$. Since the ${\bf W}_{i}$'s are identically distributed, the right hand side of \eqref{thm2-eq2} equals $4E[\{V_{n,1} - E(V_{n,1})\}^{2}]$. Note that $V_{n,1} = {\bf f}\{E({\bf S}_{{\bf W}_{1} + {\bf W}_{2}^{0} + \boldsymbol{\eta}_{n}}\mid{\bf W}_{1})\}$. Since the norm in ${\cal X}$ is assumed to be twice G\^ateaux differentiable, it follows from Thm. $4.6.15$(a) and Proposition $4.6.16$ in \citet{BV10} that the norm in ${\cal X}$ is Fr\'echet differentiable. This in turn implies that the map ${\bf x} \mapsto {\bf S}_{{\bf x}}$ is continuous on ${\cal X}\backslash\{{\bf 0}\}$ (see \citet[Corollary $4.2.12$]{BV10}). Using this fact, it follows from the dominated convergence theorem for Banach space valued random elements that
\begin{eqnarray}
E({\bf S}_{{\bf W}_{1} + {\bf W}_{2}^{0} + \boldsymbol{\eta}_{n}}\mid{\bf W}_{1} = {\bf w}_{1}) &=& E({\bf S}_{{\bf w}_{1} + {\bf W}_{2}^{0} + \boldsymbol{\eta}_{n}}) \nonumber \\
&\longrightarrow& E({\bf S}_{{\bf w}_{1} + {\bf W}_{2}^{0}}) \ = \ E({\bf S}_{{\bf W}_{1}^{0} + {\bf W}_{2}^{0}}\mid{\bf W}_{1}^{0} = {\bf w}_{1})  \label{thm2-eq3}
\end{eqnarray}
as $n \rightarrow \infty$ for almost all values of ${\bf w}_{1}$. Thus, $E(V_{n,1})$ converges to $E[{\bf f}\{E({\bf S}_{{\bf W}_{1}^{0} + {\bf W}_{2}^{0}}\mid{\bf W}_{1}^{0})\}]$ as $n \rightarrow \infty$ by the usual dominated convergence theorem. Similarly, $E(V_{n,1}^{2})$ converges to $E[{\bf f}^{2}\{E({\bf S}_{{\bf W}_{1}^{0} + {\bf W}_{2}^{0}}\mid{\bf W}_{1}^{0})\}]$ as $n \rightarrow \infty$. So, $\sum_{i=1}^{n} E[{\bf  f}^{2}\{{\boldsymbol{\Phi}}_{n}({\bf W}_{i})\}] \rightarrow {\boldsymbol{\Pi}}_{2}({\bf f},{\bf f})$ as $n \rightarrow \infty$, where ${\boldsymbol{\Pi}}_{2}$ is as defined before Thm. \ref{thm1} in Sect. \ref{2} and the expectations in that definition are taken under the null hypothesis. This completes the verification of condition ($2$) of Corollary $7.8$ in \citet{AG80}. \\
\indent Finally, for the verification of condition ($3$) of Corollary $7.8$ in \citet{AG80}, suppose that $\{{\cal F}_{k}\}_{k \geq 1}$ is a sequence of finite dimensional subspaces of ${\cal X}^{*}$ such that ${\cal F}_{k} \subseteq {\cal F}_{k+1}$ for all $k \geq 1$, and the closure of $\bigcup_{k=1}^{\infty} {\cal F}_{k}$ is ${\cal X}^{*}$. Such a sequence of subspaces exists because of the separability of ${\cal X}^{*}$. For any ${\bf x} \in {\cal X}^{*}$ and any $k \geq 1$, we define $d({\bf x},{\cal F}_{k}) = \inf\{||{\bf x} - {\bf y}|| : {\bf y} \in {\cal F}_{k}\}$. It is straightforward to verify that for every $k \geq 1$, the map ${\bf x} \mapsto d({\bf x}, {\cal F}_{k})$ is continuous and bounded on any closed ball in ${\cal X}^{*}$. Thus, using \eqref{thm2-eq3}, it follows that $\boldsymbol{\theta}(\boldsymbol{\eta}_{n}) \rightarrow 0$ as $n \rightarrow \infty$, and
\begin{eqnarray*}
\sum_{i=1}^{n} E[d^{2}\{{\boldsymbol{\Phi}}_{n}({\bf W}_{i}), {\cal F}_{k}\}] &=& 4n^{-1} \sum_{i=1}^{n} E[d^{2}\{E({\bf S}_{{\bf W}_{i} + {\bf W}' + \boldsymbol{\eta}_{n}}\mid{\bf W}_{i}) - \boldsymbol{\theta}(\boldsymbol{\eta}_{n}), {\cal F}_{k}\}] \\
&=& 4E[d^{2}\{E({\bf S}_{{\bf W}_{1} + {\bf W}_{2}^{0} + \boldsymbol{\eta}_{n}}\mid{\bf W}_{1}) - \boldsymbol{\theta}(\boldsymbol{\eta}_{n}), {\cal F}_{k}\}]  \\
&\longrightarrow& 4E[d^{2}\{E({\bf S}_{{\bf W}_{1}^{0} + {\bf W}_{2}^{0}}\mid{\bf W}_{1}^{0}), {\cal F}_{k}\}]
\end{eqnarray*}
as $n \rightarrow \infty$. From the choice of the ${\cal F}_{k}$'s, it can be shown that $d({\bf x}, {\cal F}_{k}) \rightarrow 0$ as $k \rightarrow \infty$ for all ${\bf x} \in {\cal X}^{*}$. So, we have
$\lim_{k \rightarrow \infty} E[d^{2}\{E({\bf S}_{{\bf W}_{1}^{0} + {\bf W}_{2}^{0}}\mid{\bf W}_{1}^{0}), {\cal F}_{k}\}] = 0$,
and this completes the verification of condition ($3$) of Corollary $7.8$ in \citet{AG80}. \\
\indent Thus, $\sum_{i=1}^{n} {\boldsymbol{\Phi}}_{n}({\bf W}_{i})$ converges {\it weakly} to a zero mean Gaussian distribution in ${\cal X}$ as $n \rightarrow \infty$. Further, its asymptotic covariance is ${\boldsymbol{\Pi}}_{2}$, which was obtained while checking condition ($2$) of Corollary $7.8$ in \citet{AG80}. Thus, it follows from equation \eqref{thm2-eq1} at the beginning of the proof that
\begin{eqnarray*}
 n^{1/2}\{{\bf T}_{SR} - \boldsymbol{\theta}(\boldsymbol{\eta}_{n})\} \longrightarrow G({\bf 0},{\boldsymbol{\Pi}}_{2})
\end{eqnarray*}
{\it weakly} as $n \rightarrow \infty$ under the sequence of shrinking shifts. The above weak convergence and the fact that $n^{1/2}\boldsymbol{\theta}(\boldsymbol{\eta}_{n})$ converges to ${\bf J}_{{\bf 0}}^{(2)}(\boldsymbol{\eta})$ as $n \rightarrow \infty$ complete the derivation of the asymptotic distribution of ${\bf T}_{SR}$. \\
\indent Since ${\bf T}_{S}$ is a sum of independent random elements, its asymptotic distribution can be obtained by using arguments similar to those used to derive the asymptotic distribution of $\sum_{i=1}^{n} {\boldsymbol{\Phi}}_{n}({\bf W}_{i})$ given above. It follows that $n^{1/2}{\bf T}_{S}$ has an asymptotic Gaussian distribution with mean ${\bf J}_{{\bf 0}}^{(1)}(\boldsymbol{\eta})$ and covariance ${\boldsymbol{\Pi}}_{1}$ as $n \rightarrow \infty$ under the sequence of shrinking alternatives considered in the theorem.
\end{proof}

\indent Let ${\cal X}$ be a separable Hilbert space and ${\bf Y}_{1} - {\bf X}_{1} = \sum_{k=1}^{\infty} Z_{k}{\boldsymbol{\phi}}_{k}$ for an orthonormal basis ${\boldsymbol{\phi}}_{1}, {\boldsymbol{\phi}}_{2}, \ldots$ of ${\cal X}$ and the $Z_{k}$'s are real-valued random variables. Then, the expectations defining ${\bf J}^{(1)}_{{\bf 0}}$  and ${\bf J}^{(2)}_{{\bf 0}}$ are finite if any two dimensional marginal of $(Z_{1},Z_{2},\ldots)$ has a density that is bounded on bounded subsets of $\mathbb{R}^{2}$. If ${\cal X}$ is a $L_{p}$ space for some $p \in (1,\infty)$, then ${\bf J}^{(1)}_{{\bf 0}}$ and ${\bf J}^{(2)}_{{\bf 0}}$ exist if $E(||{\bf W}_{1}||^{-1})$ and $E(||{\bf W}_{1} + {\bf W}_{2}||^{-1})$ are finite. \\
\indent We next compare the asymptotic powers of the tests based on ${\bf T}_{S}$ and ${\bf T}_{SR}$ with some paired sample mean based tests for functional data in $L_{2}[a,b]$, where $a, b \in \mathbb{R}$. \citet{CFF04} studied a test for analysis of variance in $L_{2}[a,b]$, and the test statistic for the two sample problem is based on $||\overline{{\bf X}} - \overline{{\bf Y}}||^{2}$. We consider the natural paired sample analog of this statistic and define it as $T_{1} = ||\overline{{\bf W}}||^{2}$. \citet{HKR13} studied a couple of two sample tests in $L_{2}[a,b]$ based on the projections of the sample functions onto the subspace formed by finitely many eigenfunctions of the sample pooled covariance operator. We consider the paired sample analogs of these tests, and the corresponding test statistics are defined as $T_{2} = \sum_{k=1}^{L} (\langle\overline{{\bf W}},\widehat{\boldsymbol{\psi}}_{k}\rangle)^{2}$ and $T_{3} =  \sum_{k=1}^{L} \widehat{\lambda}_{k}^{-1} (\langle\overline{{\bf W}},\widehat{\boldsymbol{\psi}}_{k}\rangle)^{2}$. Here, the $\widehat{\lambda}_{k}$'s denote the eigenvalues of the empirical covariance of the ${\bf W}_{i}$'s in descending order of magnitudes, and the  $\widehat{\boldsymbol{\psi}}_{k}$'s are the corresponding empirical eigenfunctions. Consider i.i.d. paired observations $({\bf X}_{i},{\bf Y}_{i})$, $1 \leq i \leq n$, where the ${\bf X}_{i}$'s and the ${\bf Y}_{i}$'s take values in $L_{2}[a,b]$. The next theorem gives the asymptotic distributions of $T_{1}$, $T_{2}$ and $T_{3}$ under sequences of shrinking location shift alternatives, where ${\bf W}_{i} = {\bf Y}_{i} - {\bf X}_{i}$ is symmetrically distributed about $n^{-1/2}\boldsymbol{\eta}$ for some fixed nonzero $\boldsymbol{\eta} \in {\cal X}$ and $1 \leq i \leq n$.
\begin{theorem}  \label{thm3}
(a) If $E(||{\bf W}_{1}||^{2}) < \infty$, $nT_{1}$ converges {\it weakly} to $\sum_{k=1}^{\infty} \lambda_{k}\chi^{2}_{(1)}(\beta_{k}^{2}/\lambda_{k})$ as $n \rightarrow \infty$. Here, the $\lambda_{k}$'s are the eigenvalues of the covariance operator $\boldsymbol{\Sigma}$ of ${\bf W}_{1}$ in decreasing order of magnitudes, the $\boldsymbol{\psi}_{k}$'s are the eigenfunctions corresponding to the $\lambda_{k}$'s, $\beta_{k} = \langle\boldsymbol{\eta},\boldsymbol{\psi}_{k}\rangle$, and $\chi^{2}_{(1)}(\beta_{k}^{2}/\lambda_{k})$ denotes the non-central chi-square variable with one degree of freedom and non-centrality parameter $\beta_{k}^{2}/\lambda_{k}$ for $k \geq 1$. \\
(b) Suppose that for some $L \geq 1$, we have $\lambda_{1} > \ldots > \lambda_{L} > \lambda_{L+1} > 0$. Assume that $E(||{\bf W}_{1}||^{4}) < \infty$. Then, $nT_{2}$ converges {\it weakly} to $\sum_{k=1}^{L} \lambda_{k}\chi^{2}_{(1)}(\beta_{k}^{2}/\lambda_{k})$, and $nT_{3}$ converges {\it weakly} to $\sum_{k=1}^{L} \chi^{2}_{(1)}(\beta_{k}^{2}/\lambda_{k})$ as $n \rightarrow \infty$.
\end{theorem}

\begin{proof}
(a) As in the proof of Theorem \ref{thm2}, we denote $\boldsymbol{\eta}_{n} = n^{-1/2}\boldsymbol{\eta}$, and ${\bf W}_{i}^{0} = {\bf W}_{i} - \boldsymbol{\eta}_{n}$ has the same distribution as that of ${\bf W}_{i}$ under the null hypothesis for $1 \leq i \leq n$. Now, by the central limit theorem for i.i.d. random elements in a separable Hilbert space (see \citet[Thm. $7.5$(i)]{AG80}), it follows that $n^{1/2}\overline{{\bf W}^{0}}$ converges {\it weakly} to $G({\bf 0},\boldsymbol{\Sigma})$ as $n \rightarrow \infty$. Thus, $n^{1/2}\overline{{\bf W}} = n^{1/2}(\overline{{\bf W}^{0}} + \boldsymbol{\eta}_{n})$, converges {\it weakly} to $G(\boldsymbol{\eta},\boldsymbol{\Sigma})$ as $n \rightarrow \infty$. Now, the distribution of $||G(\boldsymbol{\eta},\boldsymbol{\Sigma})||^{2}$ is the same as that of $\sum_{k=1}^{\infty} \lambda_{k}\chi^{2}_{(1)}(\beta_{k}^{2}/\lambda_{k})$ using the spectral decomposition of the compact self-adjoint operator $\boldsymbol{\Sigma}$. This proves part (a) of the proposition. \\
(b) Let ${\bf V} = (\langle\overline{{\bf W}},{\boldsymbol{\psi}}_{1}\rangle,\langle\overline{{\bf W}},{\boldsymbol{\psi}}_{2}\rangle,\ldots,\langle\overline{{\bf W}},{\boldsymbol{\psi}}_{L}\rangle)$ and $\widetilde{\boldsymbol{\beta}} = (\beta_{1},\ldots,\beta_{L})$. It follows from the central limit theorem in $\mathbb{R}^{L}$ that the distribution of $n^{1/2}({\bf V} - n^{-1/2}\widetilde{\boldsymbol{\beta}})$ converges {\it weakly} to the $L$-variate Gaussian distribution $N_{L}({\bf 0},\boldsymbol{\Lambda}_{L})$ as $n \rightarrow \infty$ under the sequence of shrinking shifts considered, where $\boldsymbol{\Lambda}_{L}$ is the diagonal matrix $Diag(\lambda_{1},\ldots,\lambda_{L})$. Thus, under the given sequence of shifts, the distribution of $n^{1/2}{\bf V}$ converges {\it weakly} to the $L$-variate Gaussian distribution $N_{L}(\widetilde{\boldsymbol{\beta}},\boldsymbol{\Lambda}_{L})$ distribution as $n \rightarrow \infty$. \\
\indent From arguments similar to those in the proof of Thm. $5.3$ in \citet{HK12}, and using the assumptions in the present theorem, we get
\begin{eqnarray}
&& \max_{1 \leq k \leq L} n^{1/2} |\langle\overline{{\bf W}},\widehat{\boldsymbol{\psi}}_{k} - \widehat{c}_{k}\boldsymbol{\psi}_{k}\rangle| = o_{P}(1)  \label{thm3-eq1}
\end{eqnarray}
as $n \rightarrow \infty$ under the sequence of shrinking shifts. Here $\widehat{\boldsymbol{\psi}}_{k}$ is the empirical version of $\boldsymbol{\psi}_{k}$ and $\widehat{c}_{k} = sign(\langle\widehat{\boldsymbol{\psi}}_{k},\boldsymbol{\psi}_{k}\rangle)$. In view of \eqref{thm3-eq1}, the limiting distribution of $n\sum_{k=1}^{L} (\langle\overline{{\bf W}},\widehat{\boldsymbol{\psi}}_{k}\rangle)^{2}$ is the same as that of $n\sum_{k=1}^{L} (\langle\overline{{\bf W}},\widehat{c}_{k}\boldsymbol{\psi}_{k}\rangle)^{2} =  n||{\bf V}||^{2}$, and the latter converges {\it weakly} to $||N_{L}(\widetilde{\boldsymbol{\beta}},\boldsymbol{\Lambda}_{L})||^{2}$ as $n \rightarrow \infty$. Thus, $nT_{2}$ converges {\it weakly} to $\sum_{k=1}^{L} \lambda_{k}\chi^{2}_{(1)}(\beta_{k}^{2}/\lambda_{k})$ as $n \rightarrow \infty$ under the sequence of shrinking shifts considered. \\
\indent It also follows using similar arguments as in the proof of Thm. $5.3$ in \citet{HK12} that under the assumptions of the present theorem, we have
\begin{eqnarray*}
&& \max_{1 \leq k \leq L} n^{1/2}\widehat{\lambda}_{k}^{-1/2}|\langle\overline{{\bf W}},\widehat{\boldsymbol{\psi}}_{k} - \widehat{c}_{k}\boldsymbol{\psi}_{k}\rangle| = o_{P}(1)
\end{eqnarray*}
as $n \rightarrow \infty$ under the given sequence of shrinking shifts. Similar arguments as in the case of $T_{2}$ yield the asymptotic distribution of $nT_{3}$, and this completes the proof.
\end{proof}

\subsection{Comparison of asymptotic powers of different tests}
\label{3.1}
\indent We now compare the asymptotic powers of the tests based on ${\bf T}_{S}$, ${\bf T}_{SR}$, $T_{1}$, $T_{2}$ and $T_{3}$ under shrinking location shifts. For this we have taken the random element ${\bf W}_{i} \in L_{2}[0,1]$ to have the same distribution as that of $\sum_{k=1}^{\infty} 2^{1/2}\{(k-0.5)\pi\}^{-1}Z_{k}sin\{(k-0.5){\pi}t\}$, where the $Z_{k}$'s are independent random variables for $k \geq 1$. We have considered two distributions of the $Z_{k}$'s, namely, the $Z_{k}$'s having $N(0,1)$ distributions, and $Z_{k} = U_{k}(V/5)^{-1/2}$ with the $U_{k}$'s having $N(0,1)$ distributions and $V$ having a chi-square distribution with $5$ degrees of freedom independent of the $U_{k}$'s for each $k \geq 1$. The latter choice is made to compare the performance of the tests based on ${\bf T}_{SR}$ and ${\bf T}_{S}$ with the mean based tests that use $T_{1}$, $T_{2}$ and $T_{3}$, when the underlying distribution is heavy tailed. The two distributions of ${\bf W}_{i}$ considered here correspond to the Karhunen-Lo\`eve expansions of the standard Brownian motion and the t process (see \citet{YTY07}) on $[0,1]$ with $5$ degrees of freedom having zero mean and covariance kernel $K(t,s) = \min(t,s)$, $t, s \in [0,1]$, respectively. We call them the sBm and the t($5$) distributions, respectively. We have chosen five degrees of freedom for the t distribution so that the finiteness of the fourth moment required in Thm. \ref{thm3} is satisfied. In a finite sample study in the next section, we will also consider t processes with one and three degrees of freedom, which violate some of the moment assumptions in Thm. \ref{thm3}. We have taken three choices of the location shift $\boldsymbol{\eta}$, namely, $\boldsymbol{\eta}_{1}(t) = c$, $\boldsymbol{\eta}_{2}(t) = ct$ and $\boldsymbol{\eta}_{3}(t) = ct(1-t)$, where $t \in [0,1]$ and $c > 0$. The plots of the shifts are given in Fig. \ref{Fig:0} below.
\begin{figure}
\begin{center}
\includegraphics[scale=0.85]{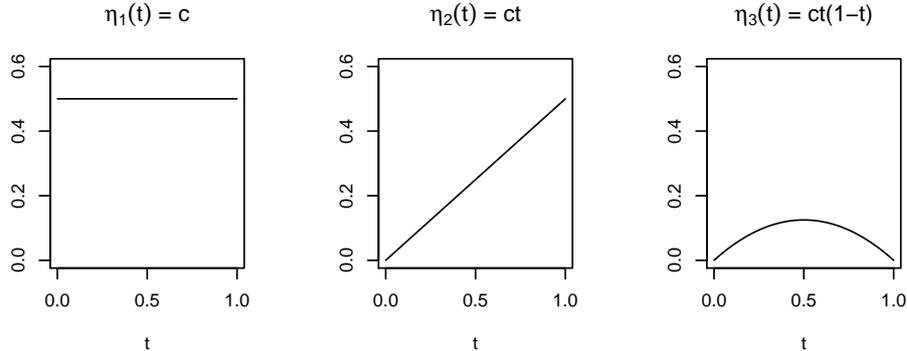}
\end{center}
\vspace{-0.2in}
\caption{Plots of the shift alternatives}
\label{Fig:0}
\end{figure}
\begin{figure}
  \begin{adjustbox}{addcode={\begin{minipage}{\width}}{\caption{%
      Plots of the asymptotic powers of the tests based on ${\bf T}_{S}$ (---), ${\bf T}_{SR}$ (- - -), $T_{1}$ (-- -- --), $T_{2}$ (-- $\times$ --) and $T_{3}$ (-- $\circ$ --) for the sBm and the t($5$) distributions under shrinking location shifts.
      }\label{Fig:1}\end{minipage}},rotate=90,center}
      \centering
      \includegraphics[height=4.3in,width=8in]{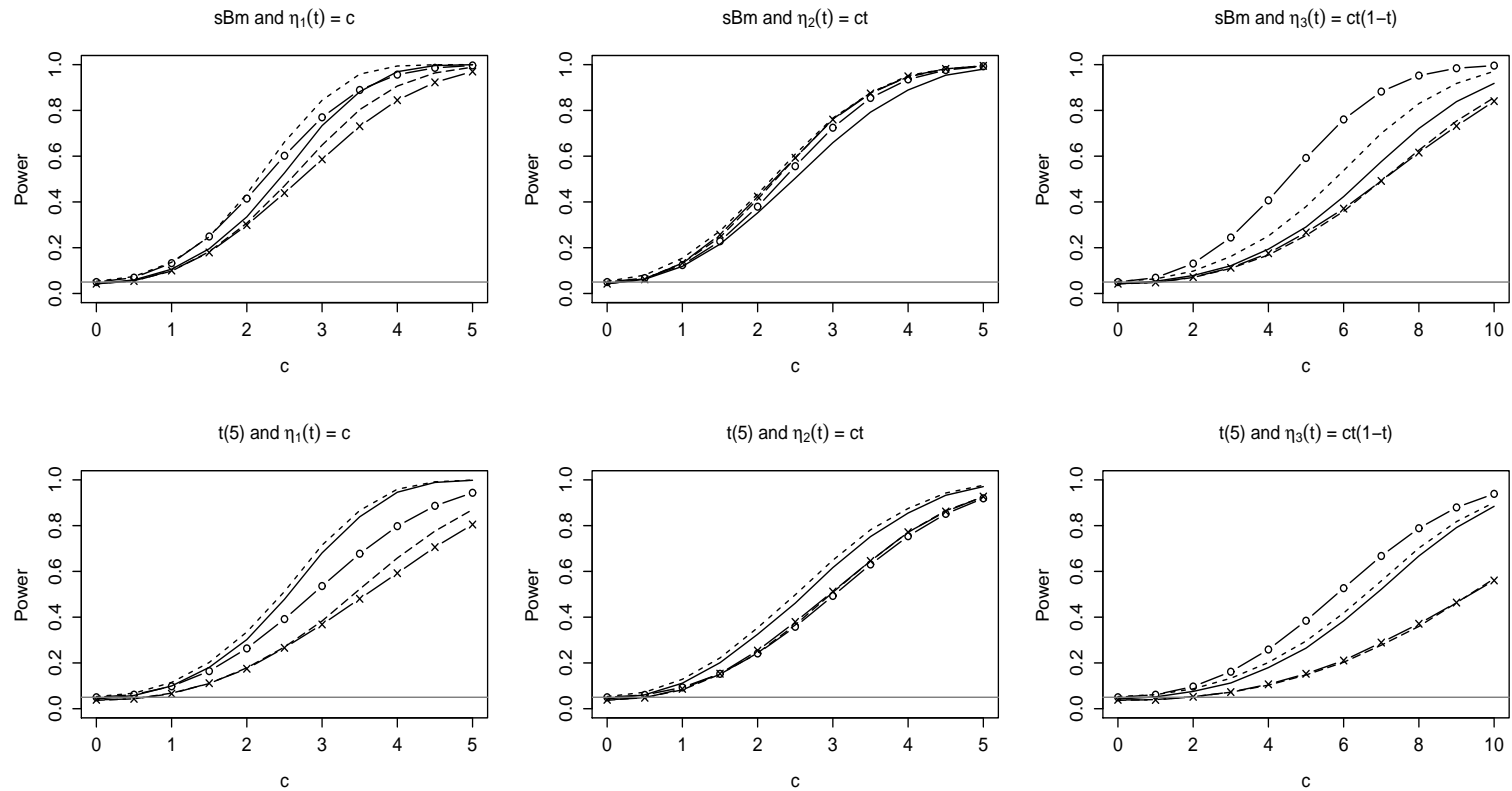}%
  \end{adjustbox}
\end{figure}

\indent For evaluating the asymptotic powers of these tests, we have used Thm. \ref{thm2} and Thm. \ref{thm3}. For each of the two distributions of ${\bf W}_{i}$, we have generated $5000$ sample functions from it. The operators ${\boldsymbol{\Pi}}_{1}$ and ${\boldsymbol{\Pi}}_{2}$ are estimated as described in Sect. \ref{2} using this sample, and the operator $\boldsymbol{\Sigma}$ is estimated by the sample covariance operator. The eigenvalues and the eigenfunctions of ${\boldsymbol{\Pi}}_{1}$, ${\boldsymbol{\Pi}}_{2}$ and $\boldsymbol{\Sigma}$ are then estimated by the eigenvalues and the eigenfunctions of the estimates of these operators. We have estimated ${\bf J}^{(1)}_{{\bf 0}}(\boldsymbol{\eta})$ and ${\bf J}^{(2)}_{{\bf 0}}(\boldsymbol{\eta})$ by their sample analogs. The asymptotic powers of the tests are then computed through $1000$ Monte Carlo simulations from the asymptotic Gaussian distributions with the estimated parameters. We have used the cumulative variance method described in \citet{HKR13} to compute the number $L$ associated with the tests based on $T_{2}$ and $T_{3}$. \\
\indent We now discuss the asymptotic powers of different tests under the sBm distribution. It is seen from Fig. \ref{Fig:1} that the tests based on ${\bf T}_{S}$ and ${\bf T}_{SR}$ asymptotically outperform the tests based on $T_{1}$ and $T_{2}$ for the shifts $\boldsymbol{\eta}_{1}(t)$ and $\boldsymbol{\eta}_{3}(t)$. The test based on ${\bf T}_{SR}$ is asymptotically more powerful than the test based on $T_{3}$ for all large values of $c$ in the case of the  shift $\boldsymbol{\eta}_{1}(t)$. However, for the shift $\boldsymbol{\eta}_{3}(t)$, the test based on $T_{3}$ is asymptotically more powerful than the test based on ${\bf T}_{SR}$. For the shift $\boldsymbol{\eta}_{2}(t)$ under the sBm distribution all the tests considered except the test based on ${\bf T}_{S}$ have similar asymptotic powers, and the latter test is asymptotically less powerful than the other competitors under the sBm distribution. \\
\indent We next consider the t($5$) distribution. The tests based on ${\bf T}_{S}$ and ${\bf T}_{SR}$ asymptotically outperform all the competing tests for all the models considered, except the test based on $T_{3}$ for the shift $\boldsymbol{\eta}_{3}(t)$. The heavy tails of the t($5$) distribution adversely effect the performance of the mean based tests, but the tests based on ${\bf T}_{S}$ and ${\bf T}_{SR}$, which use spatial signs and spatial signed ranks,  are less affected. For the shift $\boldsymbol{\eta}_{3}(t)$ under the t($5$) distribution, although the test based on $T_{3}$ asymptotically outperforms the tests based on ${\bf T}_{S}$ and ${\bf T}_{SR}$, its performance degrades significantly in comparison to its performance under the sBm distribution, which has lighter tails. \\
\indent Between the two tests proposed in this paper, the test based on ${\bf T}_{SR}$ is asymptotically more powerful than the test based on ${\bf T}_{S}$ for both the distributions and the three shift alternatives considered in this paper. However, the powers of the two tests are quite close under the t($5$) distribution.

\section{Comparison of finite sample powers of different tests}
\label{4.1}
\indent We now carry out a comparative study of the finite sample empirical powers of the tests considered in the previous section for location shift alternatives. We consider the distribution of ${\bf W}_{i}$ as in Sect. \ref{3}, i.e., ${\bf W}_{i}$ has the same distribution as that of $\sum_{k=1}^{\infty} 2^{1/2}\{(k-0.5)\pi\}^{-1}Z_{k}sin\{(k-0.5){\pi}t\}$, where the $Z_{k}$'s are independent random variables for $k \geq 1$. We have considered four distributions for the ${\bf W}_{i}$'s, namely, the sBm and the t($5$) distributions used in Sect. \ref{3} as well as the t($1$) and the t(3) distributions. For the t(1) (respectively, t(3)) distribution, $Z_{k} = U_{k}/(V/r)^{1/2}$, where the $U_{k}$'s are independent $N(0,1)$ variables and $V$ has a chi-square distribution with $r = 1$ (respectively, $r = 3$) degree of freedom independent of the $U_{k}$'s for each $k \geq 1$. The t($1$) and the t($3$) distributions are chosen to investigate the performance of the mean based tests when the moment conditions  required by them fail to hold. We have chosen $n = 20$, and each sample curve is observed at $250$ equispaced points in $[0,1]$. We consider the same three location shifts that were considered in Sect. \ref{3.1}, namely,  ${\boldsymbol{\eta}}_{1}(t) = c$, ${\boldsymbol{\eta}}_{2}(t) = ct$ and ${\boldsymbol{\eta}}_{3}(t) = ct(1-t)$ for $t \in [0,1]$ and $c > 0$. All the sizes and the powers are evaluated by averaging over $1000$ Monte-Carlo simulations. The estimated standard errors of the empirical sizes (respectively, powers) of different tests are of the order of $10^{-3}$ (respectively, $10^{-2}$ or less) for all the distributions considered. The empirical power curves of the tests are plotted in Fig. \ref{Fig:2} and Fig. \ref{Fig:3}.
\begin{figure}
  \begin{adjustbox}{addcode={\begin{minipage}{\width}}{\caption{%
      Plots of the empirical powers of the tests based on ${\bf T}_{S}$ (---), ${\bf T}_{SR}$ (- - -), $T_{1}$ (-- -- --), $T_{2}$ (-- $\times$ --) and $T_{3}$ (-- $\circ$ --) for the sBm and the t($5$) distributions.
      }\label{Fig:2}\end{minipage}},rotate=90,center}
      \centering
      \includegraphics[height=4.3in,width=8in]{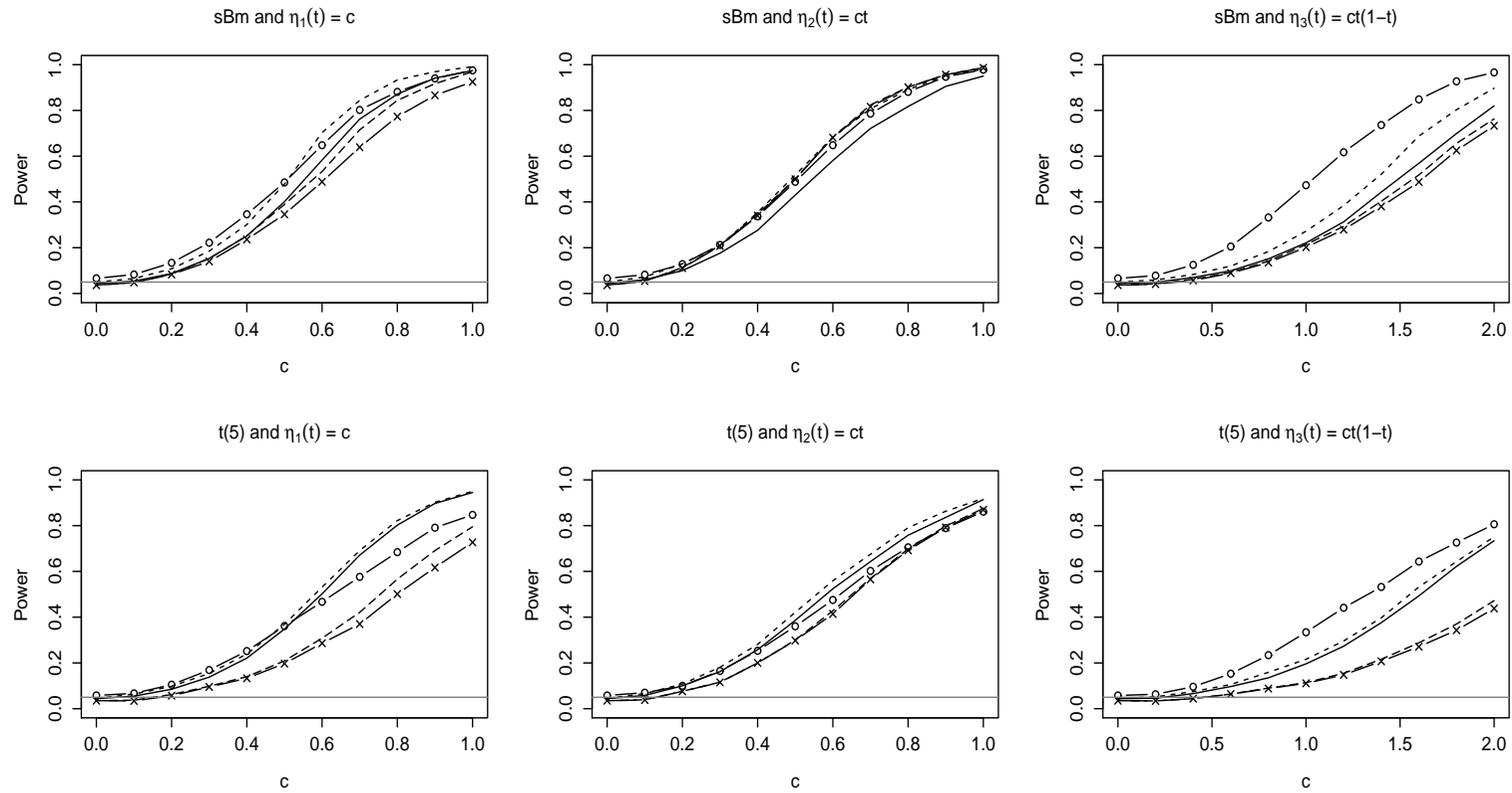}%
  \end{adjustbox}
\end{figure}

\indent For each of the tests, the difference between its observed size and the nominal $5\%$ level is statistically insignificant under the sBm, the t(3) and the t(5) distributions. However, the sizes of the tests based on $T_{1}$, $T_{2}$ and $T_{3}$ are $1\%$, $1\%$ and $2.1\%$ for the t($1$) distribution, all of which are significantly lower than the nominal level. On the other hand, the sizes of the tests based on spatial signs and signed ranks are not significantly different from the nominal level even under t(1) distribution. \\
\indent We first discuss the performance of different tests under the sBm distribution. The test based on ${\bf T}_{SR}$ is significantly more powerful than the tests based on $T_{1}$ and $T_{2}$ for the shifts ${\boldsymbol{\eta}}_{1}(t)$ and ${\boldsymbol{\eta}}_{3}(t)$. The power of the test based on ${\bf T}_{S}$ is not significantly different from that of the test based on $T_{1}$ for these two shifts although the former test has slightly more power than the latter test in our simulations. The test based on ${\bf T}_{S}$  significantly outperforms the test based on $T_{2}$ for these two shifts. The test based on ${\bf T}_{SR}$ is significantly more powerful than the test based on $T_{3}$ for large values of $c$ in $\eta_{1}(t)$. The test based on $T_{3}$ significantly outperforms the tests based on ${\bf T}_{S}$ and ${\bf T}_{SR}$ for the shift ${\boldsymbol{\eta}}_{3}(t)$. The power of the test based on ${\bf T}_{SR}$ is not significantly different from the powers of the tests based on $T_{1}$, $T_{2}$ and $T_{3}$ for the shift ${\boldsymbol{\eta}}_{2}(t)$, but the test based on ${\bf T}_{S}$ is significantly less powerful than all its competitors for this shift under the sBm distribution. \\
\indent We next consider the t(5) distribution. The tests based on ${\bf T}_{S}$ and ${\bf T}_{SR}$ are significantly more powerful than the tests based on $T_{1}$ and $T_{2}$ for all the three shift alternatives considered. The tests based on ${\bf T}_{S}$ and ${\bf T}_{SR}$ are significantly more powerful than the test based on $T_{3}$ for large values of $c$ in the shifts ${\boldsymbol{\eta}}_{1}(t)$ and ${\boldsymbol{\eta}}_{2}(t)$. Like in the case of the sBm distribution, the test based on $T_{3}$ is significantly more powerful than the tests based on ${\bf T}_{S}$ and ${\bf T}_{SR}$ for the shift $\eta_{3}(t)$. As in the asymptotic power study in Sect. \ref{3}, the performance of the mean based tests degrades significantly under the heavy tailed t($5$) distribution, while the tests based on spatial signs and signed ranks are less affected. The findings of this finite sample power study under the sBm and the t(5) distributions are very similar to those of the asymptotic power study.
\begin{figure}
  \begin{adjustbox}{addcode={\begin{minipage}{\width}}{\caption{%
      Plots of the asymptotic powers of the tests based on ${\bf T}_{S}$ (---), ${\bf T}_{SR}$ (- - -), $T_{1}$ (-- -- --), $T_{2}$ (-- $\times$ --) and $T_{3}$ (-- $\circ$ --) for the t(1) and the t($3$) distributions.
      }\label{Fig:3}\end{minipage}},rotate=90,center}
      \centering
      \includegraphics[height=4.3in,width=8in]{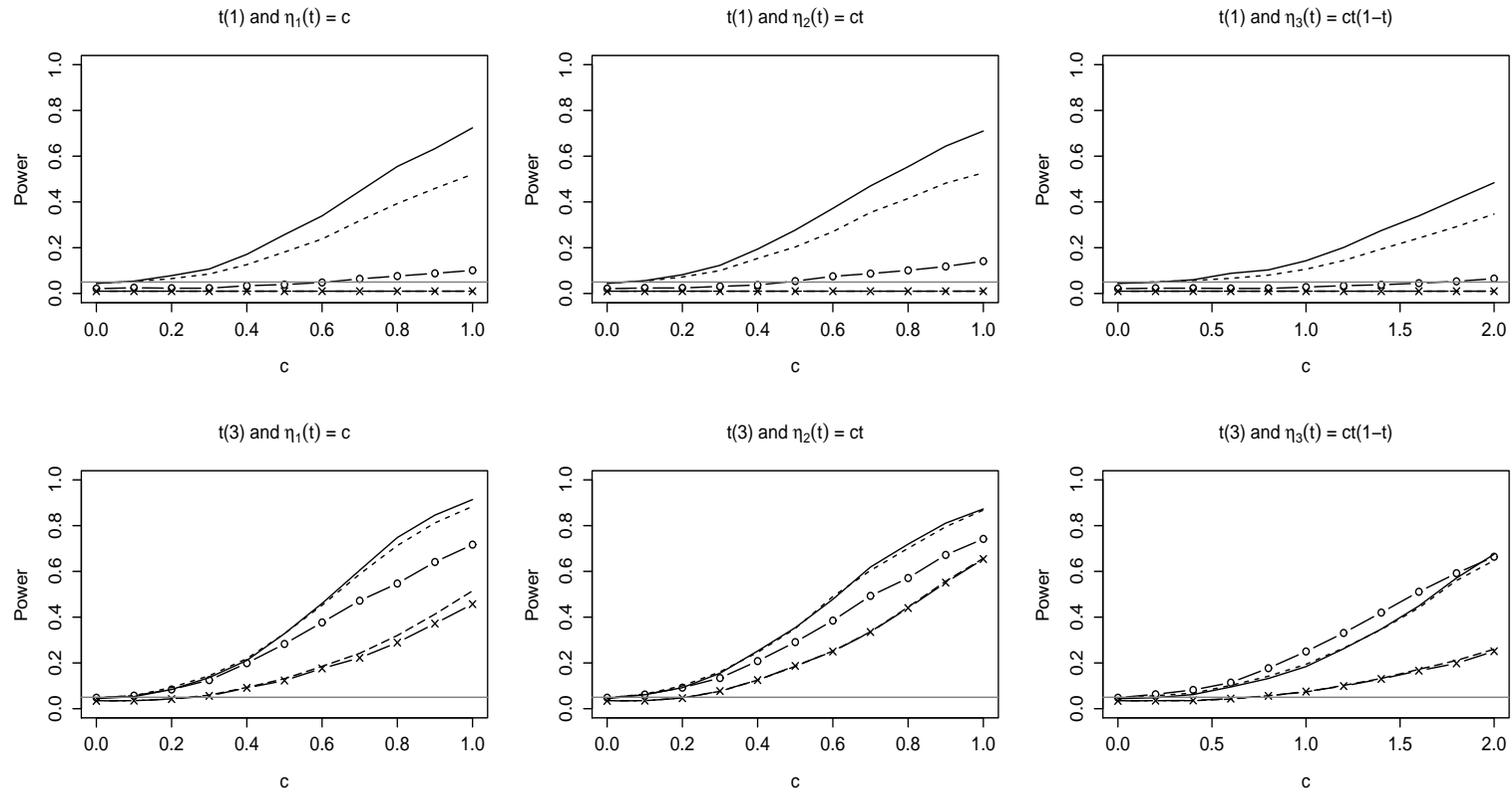}%
  \end{adjustbox}
\end{figure}

\indent The tests based on $T_{1}$, $T_{2}$ and $T_{3}$ have very low powers for all the shift alternatives considered under the t($1$) distribution and are significantly outperformed by the tests based on ${\bf T}_{S}$ and ${\bf T}_{SR}$. The non-existence of moments of the t($1$) distribution severely affects the performance of the mean based tests, but the tests based on spatial signs and signed ranks are relatively less affected. \\
\indent Under the t(3) distribution, the tests based on ${\bf T}_{S}$ and ${\bf T}_{SR}$ significantly outperform the other competitors for the shifts ${\boldsymbol{\eta}}_{1}(t)$ and ${\boldsymbol{\eta}}_{2}(t)$. The test based on $T_{3}$ has significantly more power than the tests based on ${\bf T}_{S}$ and ${\bf T}_{SR}$ only for $c \in (1,1.5)$ in the shift ${\boldsymbol{\eta}}_{3}(t)$. However, the two latter tests significantly outperform the tests based on $T_{1}$ and $T_{2}$ for the shift ${\boldsymbol{\eta}}_{3}(t)$. \\
\indent Among the two proposed tests, for all the shifts considered, the spatial sign test based on ${\bf T}_{S}$ is significantly more powerful than the spatial signed-rank test based on ${\bf T}_{SR}$ for the very heavy-tailed t(1) distribution, while it is significantly less powerful under the light tailed sBm distribution. The powers of these two tests are not significantly different for the t(3) and the t(5) distributions. However, for the t(3) distribution, the spatial sign test is slightly more powerful than the spatial signed-rank test for all the shifts considered, while the situation is reversed for the t(5) distribution. These observations are similar to the relative performance of the spatial sign and signed-rank tests for finite dimensional data under multivariate t distributions (see \citet{MOT97}).

\subsection{Robustness study of different tests}
\label{4.2}
\indent It is known that the univariate sign and signed-rank tests are robust against the presence of outliers in the data, unlike some of the mean based tests like the t test. The influence functions of the sign and signed-rank tests are bounded, and as a result, the sizes and the powers of these tests, even under moderately high contamination proportions, are not much different from their powers in the uncontaminated case (see \citet[Chapter 3]{HRRS86}). The definition of the influence function for a test as discussed in \citet[Chapter 3, p. 191]{HRRS86} can be extended to the infinite dimensional setup by using the notion of a G\^ateaux derivative. It can be shown that the influence functions of the tests based on ${\bf T}_{S}$ and ${\bf T}_{SR}$ are given by $[{\bf J}^{(1)}_{{\bf 0}}]^{-1}({\bf S}_{{\bf w}})$ and $2[{\bf J}^{(2)}_{{\bf 0}}]^{-1}\{E({\bf S}_{{\bf w} + {\bf W}_{1}})\}$ provided the inverses exist, where ${\bf w} \in {\cal X}$, and ${\bf J}^{(1)}_{{\bf 0}}$ and ${\bf J}^{(2)}_{{\bf 0}}$ are defined before Thm. \ref{thm2}. When ${\cal X}$ is a separable Hilbert space, both of these influence functions are bounded in norm under the assumptions analogous to those given in Proposition 2.1 in \citet{CCZ13}. So, it is expected that the tests based on ${\bf T}_{S}$ and ${\bf T}_{SR}$ will be robust in such cases. \\
\indent We now conduct an empirical study to assess the robustness of spatial sign and signed-rank tests proposed in this paper in comparison with the mean based tests using $T_{1}$, $T_{2}$ and $T_{3}$. Let the distribution of ${\bf Y}_{1} - {\bf X}_{1} \in L_{2}[0,1]$ be of the form $(1-\epsilon)P + {\epsilon}Q$, where $P$ is the sBm distribution considered earlier, $Q$ is the Brownian motion with covariance kernel $K(t,s) = 16\min(t,s)$, $t,s \in [0,1]$, and assume that the contamination proportion $\epsilon$ takes values $1/20$, $3/20$ and $5/20$. So, even under contamination, {\it the null hypothesis remains unchanged}. As before, the sample size is chosen to be $n = 20$, and each sample curve is observed at $250$ equispaced points in $[0,1]$. For computing the powers of the tests in the presence of outliers in the data, we have considered the location shift ${\boldsymbol{\Delta}}(t) = 0.8t$, $t \in [0,1]$. This choice ensures that the powers of the tests, even under contamination, are not too close to the nominal $5\%$ level nor too close to one, and thus a meaningful comparison between the tests can be made. The following table gives the sizes and the powers of different tests, which are evaluated by averaging over $1000$ Monte-Carlo simulations, for various contamination models considered.

\begin{table}
\caption{Sizes and powers of some tests at nominal $5\%$ level}
\label{Tab:1}
\begin{center}
\begin{tabular}{p{1cm}p{1cm}p{1cm}p{1cm}p{1cm}p{1cm}p{1cm}p{1cm}p{1cm}p{1cm}p{1cm}p{1cm}}
\hline\noalign{\smallskip}
 & \multicolumn{5}{ >{\centering\arraybackslash}p{5cm}}{Size} & & \multicolumn{3}{ >{\centering\arraybackslash}p{3cm}}{Power} \\
 \cline{2-6}\cline{8-10}\noalign{\smallskip}
 $\epsilon$ & ${\bf T}_{S}$ & ${\bf T}_{SR}$ & $T_{1}$ & $T_{2}$ & $T_{3}$ & & ${\bf T}_{S}$ & ${\bf T}_{SR}$ & $T_{3}$ \\
\noalign{\smallskip}\hline\noalign{\smallskip}
0    & 0.044 & 0.05 & 0.037 & 0.036 & 0.066 & & 0.817 & 0.891 & 0.881 \\
1/20 & 0.053 & 0.056 & 0.118 & 0.115 & 0.057 & & 0.777 & 0.818 & 0.722 \\
3/20 & 0.044 & 0.046 & 0.288 & 0.281 & 0.04 & & 0.665 & 0.616 & 0.446 \\
5/20 & 0.044 & 0.05 & 0.397 & 0.37 & 0.043 & & 0.521 & 0.423 & 0.273 \\
\noalign{\smallskip}\hline\noalign{\smallskip}
\end{tabular}
\end{center}
\end{table}

\indent It is seen from Table \ref{Tab:1} that except the tests based on $T_{1}$ and $T_{2}$, the sizes of the other tests considered are close to the nominal $5\%$ level in the contaminated as well as the uncontaminated situations. The sizes of the mean based tests that use $T_{1}$ and $T_{2}$ are much larger than the nominal level under contamination though the null hypothesis remains valid. This is probably because the asymptotic critical values of these two tests are under-estimated in the presence of contamination in the data. Interestingly, the size of the mean based test using $T_{3}$ is unaffected under contamination. It seems that the standardization involved in the statistic $T_{3}$ keeps the size of the test under control even when there are outliers in the data. \\
\indent For our power study in the presence of outliers in the data, we have excluded the tests based on $T_{1}$ and $T_{2}$ since they have very high sizes under contamination. The powers of all the tests decrease from their powers in the uncontaminated situation since the contamination increases the variability in the sample. The powers of the proposed spatial sign and signed-rank tests are significantly higher than the power of the mean based test using $T_{3}$ for all the contamination models considered. Recall that the test based on $T_{3}$ was significantly more powerful than the test based on ${\bf T}_{S}$, and its power was not significantly different from that of the test based on ${\bf T}_{SR}$ for the same location shift in the uncontaminated case. \\
\indent The powers of the tests based on spatial sign and signed rank are comparable when the contamination proportion is at most $3/20$. However, when the proportion of contamination is $5/20$, the spatial sign test becomes significantly more powerful than the spatial signed-rank test. This behaviour of the spatial sign and signed-rank tests is similar to that of the univariate sign and signed-rank tests, where it is known that the sign test is more robust than the signed-rank test for higher levels of contamination (see \citet[Chapter 3]{HRRS86}).

\section*{Concluding remarks}
\label{5}
\indent In this paper, we have studied a spatial sign test and a spatial signed-rank test for paired sample problems in infinite dimensional spaces. The tests are infinite dimensional extensions of the multivariate spatial sign and signed-rank tests considered earlier by \citet{MO95}, \citet{MOT97} and \citet{Mard99} for finite dimensional data. We have shown that the asymptotic distributions of the proposed test statistics are Gaussian after appropriate centering and scaling. It is shown that both of the proposed tests are consistent for a class of alternatives that includes the standard location shift alternatives. It is observed that under suitable sequences of shrinking location shift alternatives, the spatial sign and signed-rank tests are asymptotically more powerful than some of the mean based paired sample tests for infinite dimensional data when the underlying distribution has heavy tails. Further, even for some infinite dimensional Gaussian distributions, the spatial sign and signed-rank tests considered in this paper are more powerful than some of the mean based tests for paired sample problems. The asymptotic results are corroborated by finite sample  simulation results. \\
\indent The proposed spatial sign and signed-rank test statistics can be computed very easily, and the associated tests can be implemented using their asymptotic Gaussian distributions. The covariance operators of the asymptotic Gaussian distributions can be easily estimated as discussed in Sect. \ref{2.1}. For data in a Hilbert space, the implementations of the proposed tests become further simplified in view of the fact that in such a space, the squared norm of a Gaussian random element is distributed as a weighted sum of independent chi-square variables each with one degree of freedom.  \\
\indent The tests based on spatial signs and signed ranks studied in this paper do not require any moment assumption unlike some mean based tests for infinite dimensional data. These tests are also robust against contamination of the sample by outliers unlike the mean based tests considered in the paper. Between the two tests proposed in this paper, it is observed that the spatial sign test is better for some very heavy-tailed distributions, while the spatial signed-rank test outperforms it in some cases when the distribution has lighter tails. Further, the spatial sign test is more robust than the spatial signed-rank test, when there is a large amount of contamination in the sample.

\section*{Acknowledgement}
Research of the first author is partially supported by the Shyama Prasad Mukherjee Fellowship of the Council of Scientific and Industrial Research, Government of India.

\end{document}